\documentclass[twocolumn]{revtex4}

\usepackage{amsmath} 
\usepackage{amsfonts} 
\usepackage{amssymb} 
\usepackage{graphicx}
\usepackage{psfrag}
\usepackage{enumitem}

\newcommand{\id}{\mathbf{1}}

\newcommand{\suchthat}{\,|\,}

\newcommand{\supp}[1]{\mathrm {Supp}\,#1}
\newcommand{\synd}[1]{\mathrm {Synd}\,{#1}}

\newcommand{\stab}{\mathcal S}
\newcommand{\cent}[1]{\mathcal Z(#1)}
\newcommand{\gauge}{\mathcal G}
\newcommand{\logical}{\mathcal L}

\newcommand{\qed}{\hspace*{\fill}$\blacksquare$}
 
 \newtheorem{thm}{Theorem}
 \newtheorem{lem}[thm]{Lemma}
 
 \newtheorem{defn}[thm]{Definition}
 \newtheorem{prop}[thm]{Proposition}
 \newenvironment{proof}{\noindent \emph{Proof.}}{\qed}

\begin{document}

\title{Dimensional Jump in Quantum Error Correction}

\author{H\'ector Bomb\'in}
\affiliation{Yukawa Institute for Theoretical Physics, Kitashirakawa Oiwakecho, Sakyo-ku, 606-8502 Kyoto\\
Deparment of Mathematical Sciences, University of Copenhagen, Universitetsparken 5, DK-2100 Copenhagen \O}

\begin{abstract}
Topological stabilizer codes with different spatial dimensions have complementary properties.
Here I show that the spatial dimension can be switched using gauge fixing.
Combining 2D and 3D gauge color codes in a 3D qubit lattice, fault-tolerant quantum computation can be achieved with constant time overhead on the number of logical gates, up to efficient global classical computation, using only local quantum operations.
Single-shot error correction plays a crucial role.
\end{abstract}


\maketitle

\section{Introduction}

Quantum error correction methods~\cite{lidar:2013:quantum} that emphasize locality~\cite{dennis:2002:TQM} constitute nowadays the most promising approach for the practical implementation of a quantum computer.
In particular, topological stabilizer codes~\cite{kitaev:2003:ftanyons} receive a good deal of attention due to their flexibility and relative simplicity.
2D topological stabilizer codes are potentially easiest to implement, but low dimensionality severely constrains the operations that can be performed locally~\cite{bravyi:2013:classification}.
3D codes do not suffer from such obstructions~\cite{bombin:2007:3DCC}, but require many more qubits, among other drawbacks.
The purpose of this work is to bring together the best of the two worlds by providing a bridge between them: a procedure to switch back and forth between 2D and 3D codes.

Among 2D topological stabilizer codes color codes~\cite{bombin:2006:2DCC} are optimal in terms of the local implementation of gates.
Namely, all Clifford gates are transversal, \emph{i.e.} act individually on the physical qubits composing the code (or pair-wise for two-qubit logical gates).
See~\cite{nigg:2014:experimental} for a recent single-qubit implementation.
Unfortunately Clifford gates are not enough for universal computation, but this is all that 2D topological stabilizer codes can offer~\cite{bravyi:2013:classification,pastawski:2015:fault}.
The way out is to either resort to complementary techniques that increase the amount of resources needed~\cite{bravyi:2005:universal}, to consider more complicated codes~\cite{kitaev:2003:ftanyons}, or to increase the spatial dimension.

3D (gauge) color codes~\cite{bombin:2015:gauge} are 3D topological stabilizer codes with many remarkable characteristics that, put together,
enable fault-tolerant quantum computation with \emph{quantum-local} elementary operations, \emph{i.e.} involving only a finite depth local quantum circuit aided with global classical information processing~\cite{bombin:2015:single-shot}.
This comes at a cost: spatial locality can only be attained in 4D due to two-qubit logical gates. 
In addition, 3D color codes require $O(n^{3/2})$ qubits to correct the same number of errors
as an $n$-qubit 2D color code.

Dimensional jumps solve these problems, at least to a large extent.
As the name suggests, in a dimensional jump the spatial dimension of a local code is switched \emph{in constant time}, or more precisely via a quantum-local operation, where locality refers to a 3D layout.
In particular, the procedure allows to switch back and forth fault-tolerantly between 3D and 2D color codes.

Dimensional jumps make use of the principles of \emph{single-shot error correction}~\cite{bombin:2015:single-shot,brown:2015:fault}, \emph{i.e.} quantum-local fault-tolerant error correction.
Single-shot error correction plays a key role in the quantum-locality of operations in 3D color codes.
This is again the case for dimensional jumps, which involve error correction.

\begin{figure}
  \centering
  \includegraphics[width=6cm]{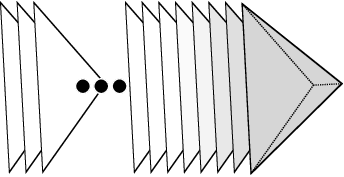}
  \caption{
  A 3D layout for fault-tolerant quantum computation.
  Each layer of the stack is a 2D color code encoding a logical qubit.
  On one extreme sits a 3D color code lattice.
  The stack acts as a memory. 
  Computations happen on the 3D end.
  }
\label{fig:layout}
\end{figure}

Equipped with dimensional jumps one can envision the 3D-local fault-tolerant quantum computing layout of figure~\ref{fig:layout}.
The starting point is a stack of 2D color codes, analogous to the one proposed for toric codes in~\cite{dennis:2002:TQM}.
Each layer encodes a single logical qubit, and all Clifford gates can be perfomed transversally.
On one extreme of the stack sits a 3D color code lattice, and the 2D color code sitting next to it can be converted back and forth into a 3D color code.
As a 2D code it can be part of the Clifford gates occurring in the stack, and as a 3D code it becomes isolated from the other 2D codes but a non-Clifford gate can be implemented, achieving universality~\cite{bombin:2007:3DCC}.

An advantage of the layout is that all logical qubits but one are encoded in 2D, dramatically reducing the resources when compared with an all-3D encoding.
Also important is that all elementary operations are quantum-local.
As a drawback, the 3D capabilities are only available in one location, and therefore parallel computation is lost
\footnote{
If the number of physical qubits is less of a concern, one can always consider a network of 3D codes connected by 2D codes.
In such a setting parallel gates are again possible, subject to the locality constraints imposed by the lattice structure.
}.
The time overhead is still constant on the number of logical gates, and for this it is enough to be able to perform swap gates in the stack, with computations confined to neighbors of the 3D code.
Finally, it is worth noting that, for the 2D and 3D constructions of~\cite{bombin:2015:gauge}, the required elementary measurements involve at most 6 physical qubits (plus any ancillas used).

\section{Background}

This section summarizes and rephrases previous results that will be needed later.

\subsection{Stabilizer codes}

A stabilizer subsystem code~\cite{poulin:2005:stabilizer} on $n$ physical qubits  is defined by two subgroups $\stab,\gauge$ of the Pauli group of operators on $n$ qubits.
The stabilizer group $\stab$ defines the code subspace where quantum information is encoded: 
encoded states are eigenstates, with eigenvalue $+1$, of all the elements of $\stab$.
The gauge group $\gauge$ generates the algebra of operators that do not disturb encoded information.
The groups $\stab$ and $\gauge$ satisfy
\begin{equation}\label{eq:stab}
-\id\not\in\stab,\qquad \stab\propto \gauge\cap\cent\gauge,
\end{equation}
where $\cent {\mathcal A}$ denotes the centralizer of $\mathcal A$ in the Pauli group.
The operators in the group $\cent\gauge$ are (bare) logical (Pauli) operators: they transform encoded states while preserving the code subspace.
Logical operators that are equivalent up to stabilizers have the same action on encoded qubits.
Therefore it is convenient to choose a representative group of logical operators $\logical\subseteq\cent\gauge$ such that each element of $\logical$ belongs to a different class of the quotient $\cent\gauge/\stab$.

\subsection{Error correction}

Error correction is the procedure that attempts to remove the errors that a code has suffered.
\emph{Ideally} it amounts to
\begin{description}[style=sameline]
\item [a.]
measure a set of generators of $\stab$ (the result is the \emph{syndrome} $\sigma$ of $\stab$), and
\item[b.]
apply a Pauli operator $E$ that yields an encoded state (such $E$ is said to have syndrome $\sigma$).
\end{description}
The operator $E$ anticommutes with the generators with negative eigenvalue outcome.
Its choice should minimize residual logical errors.

When the stabilizer generators are local the above ideal process is quantum-local.
However, in practice error correction is itself noisy, and often in making the process fault-tolerant quantum-locality is lost.
Surprisingly, for some codes quantum-locality can be preserved: they allow single-shot error correction~\cite{bombin:2015:single-shot}.

\subsection{Gauge fixing}

Gauge fixing is a procedure~\cite{paetznick:2013:universal} that allows to switch back and forth between two codes $\stab,\gauge$ and $\stab',\gauge'$ if~\cite{bombin:2015:gauge} they share a representative group of logical operators $\logical$ and
\begin{equation}\label{eq:fix1}
\stab\subseteq\stab',
\end{equation}
or, equivalently (up to a choice of signs for $\stab$ and $\stab'$)
\begin{equation}\label{eq:fix2}
\gauge'\subseteq\gauge.
\end{equation}
Any encoded state of $\stab'$ is also an encoded state for $\stab$.
Transforming an encoded state of $\stab$ into an encoded state of $\stab'$ is called gauge fixing.
The procedure is similar to error correction: \emph{ideally} it amounts to
\begin{description}[style=sameline]
\item [a.]
extract the syndrome $\sigma$ of $\stab'$, and
\item[b.]
apply some operator $E\in\gauge$ with syndrome $\sigma$.
\end{description}
The syndrome $\sigma$ is trivial for elements of $\stab$, and $E$ is unique up to elements of $\gauge'$.

\subsection{Splitting}\label{sec:split}
 
As a particular case of gauge fixing, consider that the code $\stab'$, $\gauge'$ actually consists of several codes with gauge groups $\stab_i$, $\gauge_i$ defined on \emph{disjoint} sets of $n_i$ qubits each. 
Together they form the code on $n=\sum_i n_i$ qubits
\begin{equation}
\stab'=\prod_i\stab_i, \qquad \gauge'=\prod_i\gauge_i,
\end{equation}
where it is implicitly assumed that all operators have been suitably tensored with identities to act on the $n$ qubits.
If the codes have logical representative groups $\logical_i$, we might choose for the $n$ qubit code the group
\begin{equation}
\logical:=\prod_i\logical_i
\end{equation}
Any code $\stab$, $\gauge$ on the $n$ qubits that has logical group $\logical$ and satisfies conditions~(\ref{eq:fix1}, \ref{eq:fix2}) can be gauge fixed to $\stab'$, $\gauge'$ (at least ideally).
In this case gauge fixing amounts to \emph{splitting} the $n$-qubit code into several pieces, each with some of the original physical and logical qubits.
Conversely, putting together the pieces yields an encoded state of $\stab$, $\gauge$.

\begin{figure}
  \centering
  \includegraphics[width=8cm]{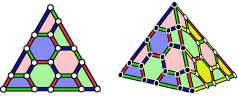}
  \caption{
(Left) A triangular 2-colex.
Plaquettes appear in their complementary color, \emph{i.e.} rg-plaquettes are colored blue.
(Right) A tetrahedral 3-colex.
Cells appear in their complementary color, \emph{i.e.} rgb-plaquettes are colored yellow.
Every facet is a triangular 2-colex.
The outer vertices (rgb-facet) are marked in black, the rest are inner vertices.
  }
\label{fig:geometry}
\end{figure}

\subsection{Colexes}

2D color codes~\cite{bombin:2006:2DCC} are defined on \emph{2-colexes}.
These are 2D trivalent lattices with 3-colored edges in which \emph{plaquettes} (2-cells) have edges of two colors.
The 2-colexes considered here are triangular.
In particular, each side of the triangle has edges in different combinations of two colors, see Fig.~\ref{fig:geometry}.
A plaquette with red and green edges is a rg-plaquette, a blue edge is a b-edge, etc.

3D color codes~\cite{bombin:2007:3DCC} are defined on \emph{3-colexes}.
These are 3D tetravalent lattices with 4-colored edges, in which plaquettes have edges of two colors and \emph{cells} (3-cells) have edges of three colors.
The 3-colexes considered here as a starting point are tetrahedral.
In particular, each facet of the tetrahedron has edges in different combinations of three colors, see Fig.~\ref{fig:geometry}.

In the present work the specific choice of 2- and 3-colexes is not relevant.
For detailed constructions and pictures, the reader is referred to the literature~\cite{bombin:2007:branyons,bombin:2007:3DCC,bombin:2015:gauge,kubica:2015:universal,brown:2015:fault}.

\subsection{Color codes}

2D color codes and 3D gauge color codes~\cite{bombin:2015:gauge} are self-dual CSS topological stabilizer codes, \emph{i.e.} the generators of the stabilizer and gauge group are products either exclusively of bit-flip $X$ operators or exclusively of phase-flip $Z$ operators, with the same geometry for $X$- and $Z$-type generators.
Therefore, the code is completely defined by the support of the generators.

Both in 2D and 3D there is one physical qubit per vertex of the colex.
Denote the respective stabilizer and gauge groups $\stab_2, \gauge_2$ and $\stab_3,\gauge_3$.
Let the support of an edge, plaquette or cell operator be the set of vertices of a given edge, plaquette or cell, respectively.
\emph{E.g.} a plaquette operator $X_p$ flips the qubits of the plaquette $p$, and so on.
Both $\gauge_2$ and $\gauge_3$ are generated by the set of all plaquette operators, something that will be highly relevant below.
The stabilizers in general depend on the geometry of the code.
In 2D triangular codes plaquette operators generate $\stab_2$, and in 3D tetrahedral codes cell operators generate $\stab_3$.
In both cases the support of $X$ and $Z$ logical operators can be chosen to be the set of all qubits.

It is interesting to mention that error correction has been substantially explored for 2D color codes~\cite{katzgraber:2009:cc,katzgraber:2009:unionjack,andrist:2010:tricolored,landahl:2011:fault,bombin:2012:strong,sarvepalli:2012:efficient,delfosse:2014:decoding,stephens:2014:efficient}, whereas for 3D little is known~\cite{brown:2015:fault}.

\section{Dimensional Jumps}

This section contains the main results of the paper.
Namely, (i) that a tetrahedral 3D gauge color code can be split into a triangular 2D color code and another 3D gauge color code with no encoded qubits and (ii) that switching fault-tolerantly back and forth between the 2D and 3D codes only requires quantum-local operations.
The generalization to higher dimensions is also briefly discussed.

\subsection{Outer and inner codes} 

The starting point is a geometrical observation: each triangular facet of a tetrahedral 3-colex is a triangular 2-colex.
Assume that such a 3-colex and a distinguished facet are given.
This facet will be denoted the \emph{outer} (2-)colex.
Conversely, the \emph{inner} (3-)colex is composed of those vertices, edges, and cells not in contact with the outer colex.
Cells/plaquettes with both inner and outer vertices are called \emph{interface} cells/plaquettes.

The 3-colex yields a 3D gauge color code $\stab_3,\gauge_3$, and the outer colex yields a 2D color code $\stab_2, \gauge_2$. 
Crucially for the results below, the 3D code admits logical operator representatives with support the set of all outer qubits, see appendix~\ref{app:codes}.
That is, the 3D gauge color code and the 2D color code have a common group $\logical$ of logical operator representatives.

A 3D gauge color code $\stab_{\text{in}},\gauge_{\text{in}}$ can also be defined for the inner colex.
Plaquette operators generate $\gauge_{\text{in}}$ (by definition) and $\stab_{\text{in}}$ is generated by (i) cell operators and (ii) the restriction to the inner qubits of interface cell operators.
The resulting code encodes no logical qubits, see Appendix~\ref{app:codes}.

\subsection{Dimensional collapse} 

As observed above, (i) the inner code has no logical qubits, (ii) the outer 2D code and the 3D code share representative logical operators and (iii)
\begin{equation}
\gauge_2\,\gauge_{\text{in}}\subseteq \gauge_3.
\end{equation}
According to section~\ref{sec:split} (with $\gauge=\gauge_3$ and $\gauge'=\gauge_2\gauge_{\text{in}}$) the 3D code splits via gauge fixing in two pieces: the outer 2D code (that keeps the logical qubit) and the inner code.
Moreover, due to the CSS structure~\eqref{eq:fix1} holds exactly, not up to a choice of signs:
\begin{equation}
\stab_3\subseteq\stab_2\,\stab_{\text{in}}.
\end{equation}

Denote by $\gauge_3|_2$ the group formed by the operators in $\gauge_3$ constrained to the outer qubits.
Its generators are edge operators, the restriction of interface plaquettes to the outer code.
Ideally, the dimensional jump from the 3D code to the 2D code amounts to
\begin{description}[style=sameline]
\item [1.]
discard inner qubits,
\item [2.]
extract the syndrome $\sigma$ of $\stab_2$, and
\item [3.]
apply some $E\in\gauge_3|_2$ with syndrome $\sigma$.
\end{description}
The operator $E$ is unique up to elements of $\stab_2$ and is a product of \emph{string operators}.
\emph{E.g.} a b-string $s$ is composed of outer b-edges $e_i$ and might have endpoints at outer rg-plaquettes, see Fig.~\ref{fig:flux}. 
The string operator $X_s$ flips the qubits of the edges $e_i$ and anticommutes with a plaquette operator $Z_p$ if and only if $p$ is an endpoint of $s$.

\begin{figure}
  \centering
  \includegraphics[width=8cm]{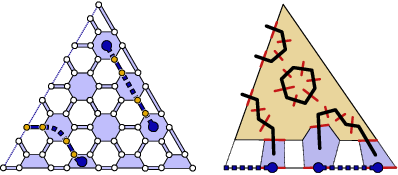}
  \caption{
  The relationship between b-strings (dotted-blue lines), rg-plaquette syndromes (blue circles) and flux lines (thick black).
  (Left)
  2-colex with its b-edges and rg-plaquettes emphasized.
  A $X_s$ string operator has support on the qubits along its string $s$. 
  It anticommutes with $Z$ plaquette operators at its endpoints.
  (Right)
  Schematic 3-colex with the outer colex on the bottom and the inner colex shaded.
  Flux lines are composed of edges dual to rg-plaquettes (in red) and can have endpoints at inner rgy-cells (blue).
  Each b-string has the syndrome of a certain flux line.
  }
\label{fig:flux}
\end{figure}

\subsection{Flux} 

The above procedure is not fault-tolerant, no matter how the syndrome $\sigma$ is extracted: even if the gauge fixing process is perfect, most pre-existing single-qubit errors on outer qubits yield a residual logical error.
To achieve fault-tolerance the key is to extract $\sigma$ \emph{indirectly} from the inner qubits.
The resulting gauge fixing process is not only fault-tolerant but also quantum-local.

Let the outer colex be the rgb-facet of the tetrahedron, \emph{i.e.} the facet with r-, g- and b-edges on it.
Suppose that inner rg-plaquette $Z_p$ operators are measured on an \emph{encoded state} of $\stab_3$.
Each plaquette has a \emph{dual} edge that pierces the plaquette connecting the centers of the cells meeting at the plaquette.
The result of the measurement is codified as the set $\gamma$ of edges dual to the inner rg-plaquettes with eigenvalue -1.

Given a cell $c$ with rg-plaquettes $p_i$, every vertex in $c$ belongs to exactly one of the $p_i$. 
Therefore
\begin{equation}\label{eq:closed}
Z_c=\prod_{i} Z_{p_i}.
\end{equation}
For the initial state $Z_c=\id$ and thus $\gamma$ has no \emph{inner endpoints}, \emph{i.e.} inner cells at which an odd number of edges meet.
Therefore $\gamma$ is a disjoint union of paths $\gamma_i$, or `flux-lines', that are either closed or have endpoints (i) at the rgy-facet or (ii) at an interface rgy-cell, see Fig.~\ref{fig:flux}.

Every interface rgy-cell $c$ has a unique outer rg-plaquette $p_c$ (the rest are inner, unless the colex is pathological), see Fig.~\ref{fig:flux}.
A flux-line is said to have an endpoint at $p_c$ when it has an endpoint at $c$.
According to~\eqref{eq:closed} $Z_{p_c}$ has eigenvalue $-1$ if and only if $p_c$ is the endpoint of and odd number of flux-lines $\gamma_i$.
Thus the syndrome of $\stab_2$ can be recovered from the measurement of inner rg-, gb- and rb-plaquette $X$ and $Z$ operators, which commute, see appendix~\ref{app:codes}~\footnote
{
It is worth emphasizing that the gauge fixing discussed here differs from the of one considered in~\cite{bombin:2015:gauge}, which involved all plaquette operators of either $X$ or $Z$ type, which do not commute.
}.
The steps 1 and 2 above can be substituted by:
\begin{description}[style=sameline]
\item [1+2.]
Obtain a syndrome $\sigma$ from the destructive measurement of the plaquette operators in $\gauge_{\text{in}}$ with colors matching the outer plaquettes.
\end{description}

\subsection{Measurement errors}

Consider again the measurement of inner rg-plaquette $Z_p$ operators.
Suppose that the original encoded state is noiseless but measurements can fail: instead of the correct dual edge set $\gamma$ they yield $\gamma+\delta$, with $+$ the symmetric difference of sets: plaquette operators corresponding to edges in $\delta$ are assigned the wrong eigenvalue.
The set $\gamma+\delta$ can have inner endpoints.
Let $\delta_0$ be a set of dual edges of minimal cardinality with the same inner endpoints (efficiently computable using perfect matching~\cite{dennis:2002:TQM}).
It provides an `effective' set
\begin{equation}
\gamma_\mathrm{eff}:=\gamma+\delta+\delta_0=\gamma +\bigsqcup_i \omega_i,
\end{equation}
where $\delta+\delta_0$ decomposes as a disjoint union of flux-lines $\omega_i$ (because $\delta+\delta_0$ has, like $\gamma_\mathrm{eff}$, no inner endpoints).

For every outcome $\gamma$ of rg-flux-lines there is some operator $E_{\gamma}$ that is a product of b-string $X$ operators and has the syndrome corresponding to $\gamma$; it is unique up to stabilizers.
By using $\gamma_\mathrm{eff}$ as input for the third gauge fixing step, instead of the correct operator $E_\gamma$ we apply
\begin{equation}
E_{\gamma_\mathrm {eff}}\sim E_{\gamma} \,E_{\delta+\delta_0} \sim E_\gamma\,\prod_i E_{\omega_i}
\end{equation}
where the equivalence is up to stabilizers: measurement noise $\delta$ translates into errors $E_{\omega_i}$ at the final stage.

How bad are these errors? Each $E_{\omega_i}$ is a b-string operator $X_{s_i}$ that we can \emph{choose} subject to the constraint that its endpoints on outer plaquettes should match those of $\omega_i$. 
It follows by inspection of the different geometries, depicted in Fig.~\ref{fig:flux}, that the number of qubits in the support of $E_{\omega_i}$ is, up to a constant depending on the lattice structure, smaller than the length $|\omega_i|$ of the flux-line.
Moreover, at least half of the edges of $\omega_i$ belong to $\delta$, rather than $\delta_0$: 
\begin{equation}
|\delta_0+\omega_i|= |\delta_0|-|\omega_i\cap\delta_0|+|\omega_i\cap\delta|\geq |\delta_0|,
\end{equation}
because $\delta_0+\omega_i$ has the same inner endpoints as $\delta_0$.
Thus an error $E_{\omega_i}$ of large support requires a large $\omega_i$ within which at least $|\omega_i|/2$ measurements fail.
This suggests that local noise in the measurement process will yield local residual noise.
Indeed, a standard argument~\cite{dennis:2002:TQM,bombin:2015:single-shot} shows that if the noise is local and below a threshold, so will be the residual noise, see appendix~\ref{app:gauge fixing}.

\subsection{Fault tolerance} 

In general the original 3D state will be noisy, and so will be the measurements and the application of $E\in\gauge_3|_2$.
Local errors affecting outer qubits at any time will remain local, because the application of $E$ is local.
Local errors affecting inner qubits and previous to measurements can be absorbed as local measurement errors.
If all the noise is local and below a threshold, so will be the residual noise after the `dimensional collapse'.

\subsection{Blowing up} 

The inverse dimensional jump only requires initializing the inner code.
Since it encodes no logical qubits, it suffices to apply error correction to an arbitrary state of the inner qubits.
Moreover, since it is a 3D gauge color code with local stabilizer and gauge generators it admits single-shot error correction, see appendix~\ref{app:codes},
and thus the process is quantum-local.

\subsection{Higher dimensions}

Colexes and gauge color codes can be defined for arbitrary dimensions~\cite{bombin:2015:gauge}.
For a given $D$-colex it is possible to build different color codes with labels $(d,e)$ that indicate the dimension of the gauge generators: $d$ and $e$ are positive integers with $d+e\leq D$, $X$-type generators are $(e+1)$-cell operators, and $Z$-type generators are $(d+1)$-cell operators.

The most interesting class of color codes is that constructed out of \emph{simplicial} colexes, which generalize the triangular and tetrahedral colexes considered above.
They encode a single qubit, and the logical $X$ and $Z$ operators can be chosen to have as support the whole colex.
For these codes gauge fixing can be used to change, within a given colex, the parameters $(d,e)$ at will~\cite{bombin:2015:gauge}.

The dimensional jump described above switches between a given $(1,1)$ tetrahedral color code and a $(1,1)$ triangular color code defined on any of the facets of the tetrahedron.
Analogously, one can switch between a $(d,e)$ $D$-simplicial color code and a $(d,e)$ $(D-1)$-simplicial color code defined on any of the facets of the $D$-simplex~\footnote
{
Indeed, the gauge generators trivially match, and logical operators can be chosen to have as support the whole facet.
To check that such operators commute with any cell operator of dimension at least two, observe that there exists only one color $\kappa$ such that $\kappa$-edges might have a single vertex in the facet, and the cell operator contains edges of more than one color.
}.
The lesson is that gauge color codes with different values of $D$ or $(d,e)$ are more than just separate codes.
Altogether they form a \emph{system} of topological stabilizer codes, and much more is possible by making them work together than by using them separately.

Finally, for every dimension $D$ there exists a minimal simplicial colex with $2^{D+1}-1$ vertices~\cite{bombin:2015:gauge}.
The corresponding color codes are quantum Reed-Muller codes and are known to be related via gauge fixing~\cite{anderson:2014:fault}.

\section{3D-local computation}

Dimensional jumps open the door to a 3D-local fault-tolerant quantum computer in which all operations are quantum local.
This sections describes a particular approach to achieve this this.
The time overhead is constant on the number of logical gates, and only a single 3D color code lattice is required.
As a drawback, any parallelism of the original circuit is lost, so no higher-level fault-tolerance is possible.

\section{Layout}

The general layout is described in figure~\ref{fig:layout}: a stack of 2D color codes with a 3D color code at one end of the stack.
Computations are performed at the end where the 3D code lives, whereas the stack acts as a memory.
For the logical qubits encoded on the 2D color codes next to the 3-colex structure, quantum-local initialization and universal gates are available via dimensional jumps (quantum-local measurements do not require them~\cite{dennis:2002:TQM}).
Moreover, for those 2D codes single-shot error correction is available via dimensional jumps, making fault-tolerant CNot gates quantum-local.
The only operation required in the stack is the (trivially transversal) swapping of the 2D codes.
This allows to move logical qubits so that they are available at the computing end of the stack when needed. 
Single-shot error correction is not available in the stack, and it is not needed. 
For 2D codes in the stack error correction amounts to keeping track of errors by repeatedly measuring the stabilizer generators~\cite{dennis:2002:TQM}.

Only one obstacle is left: ensuring that the required logical qubits are always available at the computing end for each step of the computation, without incurring in waiting times.
A procedure to achieve this is given next.

\subsection{Computing at the end of a stack}\label{sec:computing}

Consider, as described above, a quantum computer where logical qubits are placed on a stack and non-trivial computations are limited to the end of the stack: the rest of qubits can only be swapped with their neighbors in the stack.
It is possibly not entirely obvious that the number of rounds of parallel gates can be, up to a constant, equal to the number of gates in the circuit model.
This section provides a simple algorithm for swapping the stacked logical qubits that achieves this.

At the end of the stack there are a number of (logical) qubits on which non-swap gates can be performed.
At least there should be two of them, since two-qubit gates are necessary.
These qubits can be regarded as the `internal' memory of a `processing unit', while the rest form an `external' memory: the stack.
On the stack two kinds of operations are allowed: the swap of neighboring qubits and the swap of the topmost qubit with an internal qubit.

The computation is divided in steps.
At each step a certain external qubit is required to be at the end of the stack, so that it can be accessed by the processing unit.
The challenge is to perform on the stack, after each computational step, a finite depth circuit composed of nearest neighbor swaps that places the next qubit to be processed at the end of the stack.

The proposed algorithm is the following.
The positions in the stack are labeled with integers: position 1 is the topmost.
At the $s$-th step the position $i$ has an integer label $m_i$: it indicates that the qubit currently at position $i$ should be at position $1$ on the step $m_i\geq s$, but not before.
In particular at the $s$-th step
\begin{equation}\label{eq:condm1}
m_1=s.
\end{equation}
When a qubit will not be used again this label can take any value larger than the number of steps, with the condition that $m_i\neq m_j$ for $i\neq j$.

At start, before the first step, qubits are ordered according to their first use, \emph{i.e.}
\begin{equation}
m_i<m_{i+1}, \qquad m_1=1.
\end{equation}
After the step $s$ is performed the label $m_1$ is updated to its new value (the rest stay clearly the same) and the following operations are performed, each consisting of a round of parallel swaps:
\begin{description}[style=sameline]
\item [1.]
For every $n\geq 0$, if
\begin{equation}
m_{2n+1} > m_{2n+2}
\end{equation}
swap the qubits at positions $2n+1$ and $2n+2$ (and the labels $m_{2n+1}$ and $m_{2n+2}$ accordingly).
\item [2.]
For every $n\geq 1$, if
\begin{equation}
m_{2n} > m_{2n+1}
\end{equation}
swap the qubits at positions $2n$ and $2n+1$ (and the labels $m_{2n}$ and $m_{2n+1}$ accordingly).
\end{description}

To check that \eqref{eq:condm1} is satisfied at every step, it suffices to show that at every step and for all positions $i$
\begin{equation}\label{eq:condm1b}
m_1\leq m_i.
\end{equation}
Suppose that after the $s$-th step and the update of $m_1$
\begin{equation}\label{eq:condm2}
m_{2n} < m_{2n+k},\qquad n,k>0.
\end{equation}
Then, after the first round of swaps is performed 
\begin{equation}\label{eq:condm3}
m_{2n+1} < m_{2n+1+k},\qquad n\geq 0, k>0,
\end{equation}
and after the second round of swaps is performed both~\eqref{eq:condm1b} and~\eqref{eq:condm2} hold again. 
In particular $m_1=s+1$ as required.
Thus the inequalities~\eqref{eq:condm2} are an invariant of the two-round procedure. 
Since they are initially satisfied, the algorithm works.

\begin{acknowledgments} 

I thank the MINECO grant FIS2012-33152, the CAM grant QUITEMAD+, 
the Sapere Aude grant of the Danish Council for Independent Research, the ERC Starting Grant QMULT and the CHIST-ERA project CQC.
This work was supported by the International Research Unit of Advanced Future Studies at Kyoto University.

\end{acknowledgments}

\appendix

\section{Boundaries in 3D color codes}\label{app:codes}

The main text makes use of some basic properties of gauge color codes that were presented in~\cite{bombin:2015:single-shot}.
For easier reference they are gathered in this appendix.

\subsection{Geometry}

The tetrahedral colexes of the main text are just an example of a larger class of geometries for color codes.
As in~\cite{bombin:2015:single-shot}, of interest here are 3-colexes $\Lambda$ that are topological balls and are obtained by removing some vertices (and all cells in contact with them) of a larger colex $\bar\Lambda$ without boundary.
The boundary of $\Lambda$ is a topological sphere divided in \emph{regions} (discs), that meet at  \emph{borders} (open lines), that meet at \emph{corners} (points).
They are defined as follows.
Consider a rgb-cell of $\bar\Lambda$ that is not part of $\Lambda$ but is in contact with $\Lambda$.
The rgb-cell and $\Lambda$ share a set of plaquettes $R$, which can only contain rg-, rb- or gb-plaquettes.
Typically these plaquettes form a topological disc in the boundary of $\Lambda$, which is called a rgb-region. (otherwise there are several discs, each a region).
The r- and g-edges that separate a rgb-region and a rgy-region are said to form a rg-border.
For each rg-border there is exactly one rg-plaquette in $\bar\Lambda$ (and not in $\Lambda$) that contains its edges.
Finally, when a rg-border, a rb-border and a gb-border meet at a vertex, this is said to be a y-corner of the colex.
Again, for each y-corner there is exactly one y-edge in $\bar\Lambda$ (and not in $\Lambda$) that contains it.
This description of the boundary of $\Lambda$ in terms of regions, borders and corners does not depend on $\bar\Lambda$: it is intrinsic to $\Lambda$.

In the case of the tetrahedral colex each facet is a region with a different color combination: rgb, rgy, rby and gby.
The simplest $\bar\Lambda$ is obtained by adding a single additional vertex together with a single edge, plaquette and cell for each color combination.
The resulting manifold is a 3-sphere.
Due to this construction, color codes in tetrahedral colexes have been called `punctured'.

The inner colex described in the main text is another example of the above general class of geometries.
It is obtained from a tetrahedral colex by erasing all the cells in contact with one of the facets/regions. 
Each of the removed cells contributes a region in the inner colex.
\emph{E.g.} in the main text the outer set of vertices is the rgb-facet, which is in contact with rgy-, rby- and gby-cells.
To obtain the inner colex all the outer vertices are removed, together with those cells.
The remaining inner colex has rgy-, rby- and gby-regions, which are of two kinds.
They can correspond to one of the original facets of the tetrahedron, or they can correspond to one of the erased cells.

The regions of the tetrahedral and inner colex turn out to have very different properties at the level of the code.
This motivates some definitions.
A border is \emph{odd} if it connects two corners with different colors.
A region without odd borders is said to be \emph{frozen}.
A region with an odd number of odd borders of any given color is said to be \emph{free}.
In a tetrahedral colex all regions are free, \emph{e.g.} the rgb-facet has a single odd rg-border, connecting the r-corner and the b-corner.
In an inner colex all regions are frozen, because all corners have the same color (in the above example they are y-corners).

\subsection{Stabilizer and logical operators}

As stated in the main text, a 3D gauge color code is obtained by placing a qubit at each vertex of the 3-colex and attaching gauge $X$ and $Z$ generators to plaquettes.
A bit-flip plaquette operator $X_p$ has support on the qubits belonging to the plaquette $p$, and similarly for a phase-flip plaquette operator $Z_p$.
If $p$ is a rg-plaquette then $X_p$ can only anticommute with $Z_{p'}$ if $p'$ is a by-plaquette:
\emph{e.g.} if $p'$ is a gb-plaquette then it shares with $p$ a certain number $x$ of b-edges, and then exactly $2x$ qubits (a qubit belongs at most to one b-edge).

For the above class of geometries, it is shown in~\cite{bombin:2015:single-shot} that the generators of $\cent\gauge$ are (i) cell operators and (ii) region operators, \emph{i.e.} operators of the form
\begin{equation}
X_R:=\prod_{i\in R} X_i,\qquad Z_R:=\prod_{i\in R} Z_i,
\end{equation}
where $X_i$, $Z_i$ are the Pauli $X$, $Z$, operators on the $i$-th qubit and $R$ is a region, regarded here as set of vertices/qubits.
Moreover, (i) a region $R$ is free if and only if $X_R$ and $Z_R$ anticommute, and (ii) given two different regions $R$ and $R'$, $X_R$ and $Z_{R'}$ anticommute if and only if they share an odd number of odd borders.
From this result it follows that to obtain a representative group $\logical$ of bare logical operators 
\begin{equation}
\logical\subseteq \cent\gauge, \qquad \logical\simeq \cent\gauge/\stab,
\end{equation}
it suffices to choose a minimal generating set for $\logical$ among $X$ and $Z$ region operators.
It follows also that for frozen regions $R$ the operators $X_R$ and $Z_R$ are stabilizer elements.
In particular, if $R$ is a frozen rgy-region with $y$-corners $X_R$ can be obtained as a product of plaquette operators $X_p$ as follows (and similarly for $Z$ operators)
\begin{equation}\label{eq:stabR}
X_R = \prod_{p \in \text{rg-plaquettes in $R$}} X_p.
\end{equation}
Indeed, every vertex/qubit belongs to exactly one rg-plaquette of $\bar\Lambda$, and thus either (i) it belongs to exactly one rg-plaquette of $\Lambda$ or (ii) it is part of an rg-border.
The second options is never true for a frozen rgy-region with y-corners: it only shares vertices with ry-, gy and by borders.
Moreover, $R$ only borders with rby- and gby-regions, which cannot contain rg-plaquettes: if a vertex belongs to $R$, it belongs to a unique rg-plaquette contained in $R$.

When all the regions are free region operators cannot contribute any stabilizer generators: a product of $X$ and $Z$ region operators can only belong to the stabilizer if it is trivial.
When all the regions are frozen, on the other hand, all region operators belong to the stabilizer.
Region operators are, however, not independent.
In a colex where all corners are y-corners the following identity holds (and similarly for $Z$ operators)
\begin{equation}\label{eq:redundant}
\prod_{c \in\text{rgb-cells}} X_c \prod_{c \in\text{rgy-cells}} X_c =\prod_{R \in\text{rgy-regions}} X_c,
\end{equation}
where $X_c$ are bit-flip cell operators and $X_p$ are bit-flip plaquette operators.
This is because each vertex/qubit belongs (i) exactly to one rgb-cell and (ii) exactly to one rgy-cell, \emph{unless} it belongs to a rgy-region.
Ulternatively, this lack of independence of region operators is a trivial consequence of the flux picture of the main text and the expression~\eqref{eq:stabR}: a flux incoming at a certain rg-plaquette of a given rgy-region has to exit at another rg-plaquette of another (or the same) rgy-region.

In tetrahedral codes all the regions are free and thus the stabilizer $\stab$ is generated by cell operators alone.
In inner codes all the regions are frozen and thus the stabilizer has as generators (i) cell operators and (ii) region operators, except those coming from the facets of the original tetrahedral colex (facet operators are redundant due to~\eqref{eq:redundant}).
The take home message is that for all the geometries considered in the main text there exist a set local generators of the stabilizer.

\subsection{Single-shot error correction}

Gauge color codes with local stabilizer and gauge generators exhibit single-shot error correction~\cite{bombin:2015:single-shot}.
In particular, this is true for the above geometries when (i) all the regions are free or (ii) all the regions are frozen.
The quantum-local process for fault-tolerant error correction can be split in two analogous processes that correct separately bit-flip and phase-flip errors.
Let $\gauge_Z$, $\stab_Z$ denote the gauge and stabilizer subgroups generated by $Z$ operators.
Single-shot error correction of bit-flip errors amounts to
\begin{description}[style=sameline]
\item [1.]
measure $Z$ plaquette operators (generators of $\gauge_Z$), 
\item[2.]
process the measurements (classically) to obtain a syndrome $\sigma$ for $\stab_Z$ (a subgroup of $\gauge_Z$), 
\item[3.]
choose a bit-flip operator $E$ with syndrome $\sigma$, and
\item[4.]
apply $E$. 
\end{description}

In the gauge color code families considered in the main text both the stabilizer and gauge group have local generators. 
Single-shot error correction is both possible for tetrahedral codes (because all regions are free) and for inner codes (because all regions are frozen).

\section{Single-shot gauge fixing}\label{app:gauge fixing}

The purpose of this appendix is to show that the dimensional collapse procedure of the main text is fault-tolerant.
The derivation will be entirely parallel to the one used in~\cite{bombin:2015:single-shot} to show that single-shot error correction is feasible with 3D gauge color codes.
We will need in particular the following definition and lemma from~\cite{bombin:2015:single-shot}.

\begin{defn}\label{def:bounded}
Let $p(A)$ be a probability distribution over subsets $A\subseteq B$ of some set $B$.
Given some $\alpha>0$ the distribution $p$ is $\alpha$-bounded if for every $A\subseteq B$
\begin{equation}\label{eq:tilde}
\tilde p(A):=\sum_{A'\supseteq A} p(A') \leq \alpha^{|A|}.
\end{equation}
\end{defn}

\begin{lem}\label{lem:bounded}
Consider a family of graphs with bounded maximum degree and some $k>0$.
Each graph has a set of nodes $\Gamma$ and comes equipped with (i) two subsets of nodes $\Gamma_i\subseteq \Gamma$, $i=1,2$, and (ii) a function $f:\Gamma_1\longrightarrow\Gamma$ such that for every cluster (set of nodes) $V\subseteq \Gamma_1$ 
and for every connected component $V_{\text c}$ of the cluster $f(V)$
\begin{equation}\label{eq:connectivity}
|V\cap V_{\text c}|\geq k |V_{\text c}|.
\end{equation}

There exists some $\alpha_0>0$ such that for every $\alpha$ with $0<\alpha<\alpha_0$ the following holds.
If a probability distributions $p_1(V)$ over clusters $V\subseteq\Gamma_1$ is $\alpha$-bounded, then the probability distribution $p_2(V')$ over clusters $V'\subseteq\Gamma_2$ defined by
\begin{equation}\label{eq:p2}
p_2(V'):=\sum_{V\suchthat f(V)\cap\Gamma_2= V'} p_1(V).
\end{equation}
is $\beta$-bounded with
\begin{equation}\label{eq:beta}
\beta :=  \frac {(\alpha/\alpha_0)^{k}}{1-(\alpha/\alpha_0)^{k}},
\end{equation}
\end{lem}

As argued in the main text, it is enough to show that local measurement noise in the syndrome extraction step gives rise to local residual noise after the gauge fixing step.
We will do this for some simple but compelling enough model of noise.
The starting point is some family of 3D tetrahedral gauge color codes that is regular enough,~\emph{i.e.} the lattice is translationally invariant up to boundaries and locally identical across the family, and the shape of the tetrahedron just scales in size across the family, without other changes.
As in the text, we will focus on the gauge fixing of rg-plaquette $Z$ operators in $\stab_2$.
Other color combinations and $X$ operators are analogous.

Let $\mathcal E$ be the set of edges dual to inner rg-plaquettes.
Recall that each edge $e\in\mathcal E$ has two endpoints, each of which can be (i) in an inner cell, in which case it is called an \emph{inner endpoint}, (ii) in an interface rgy-cell, in which case it is regarded as an endpoint at the corresponding outer rg-plaquette, and called \emph{outer endpoint}, and (iii) at the rgy-facet (this case is uninteresting).
The endpoints of some $\gamma\subseteq\mathcal E$ are those inner cells or outer plaquettes that are endpoints of an odd number of edges $e\in\gamma$.
Let $\mathcal E_0\subseteq \mathcal E$ be the subset of elements without inner endpoints.
Given $\gamma\in\mathcal E_0$ let $\partial \gamma$ denote the set of outer rg-plaquettes where $\gamma$ has an endpoint.

Let $\mathcal B\subset\gauge_3|_2$ be the set of bit-flip operators that are products of b-string operators (in the outer 2D color code).
A given $E\in \mathcal B$ can only anticommute with an outer plaquette operator $Z_p$ if $p$ is an rg-plaquette.
Let $\synd E$ denote the set of such rg-plaquettes, so that $\synd E$ characterizes the syndrome of $E$ in $\stab_2$.
The geometry of a 2D triangular color code is such that $\synd E$ can take any value, it is unconstrained.
In particular, for every $\gamma\in\mathcal E_0$ there exists some $E_\gamma\in\mathcal B$ with
\begin{equation}
\synd E_\gamma = \partial\gamma.
\end{equation}
The operator $E_\gamma$ is unique up to stabilizers, and we choose it so that it has minimal support, \emph{i.e.} for every $E\in\mathcal B$
\begin{equation}\label{eq:Egamma}
\synd E=\partial\gamma\quad\Longrightarrow\quad |\supp E_\gamma|\leq|\supp E|. 
\end{equation}
As discussed in the text, it follows by inspection of the geometry of the problem that there exists some constant $K>0$ such for every set of dual inner edges $\gamma$ without inner endpoints
\begin{equation}\label{eq:K}
|\supp E_{\gamma}|\leq K |\gamma|.
\end{equation}
In particular, $K$ is constant across the family of codes.

We model local noise in the measurements through a probability distribution $p(\delta)$ over sets $\delta\subseteq\mathcal E$.
Recall that plaquette operators from rg-plaquettes dual to the elements in $\delta$ give a wrong measurement outcome.
To model the locality of the measurement noise, we impose that $p(\delta)$ is $\alpha$-bounded for some $\alpha>0$.

The residual error when the wrong set of measurement outcomes is given by $\delta\subseteq \mathcal E$ is 
\begin{equation}
E_{\delta+\delta_0},
\end{equation}
where $\delta_0\subseteq \mathcal E$ has the same inner endpoints as $\delta$ and minimal cardinality among sets with that property, \emph{i.e.} for any $\delta'\in\mathcal E$,
\begin{equation}\label{eq:delta0}
\partial\delta=\partial\delta'\quad\Longrightarrow\quad |\delta_0|\leq|\delta'|. 
\end{equation}
Let $\mathcal Q$ be the set of physical qubits in the outer 2D code.
To quantify the residual noise we construct a distribution $p'(Q)$ over sets of qubits $Q\subseteq\mathcal Q$, with $p'(Q)$ the probability that the support of $E_{\delta+\delta_0}$ is $Q$, \emph{i.e.}
\begin{equation}\label{eq:pp}
p'(Q):=\sum_{\delta| \supp{E_{\delta+\delta_0}}=Q}p(\delta).
\end{equation}

The following result quantifies the locality of the residual noise.
It shows in particular that it can be made as small as desired by improving the quality of the measurements.
\begin{prop}
Given a family of 3D gauge color codes as described above and satisfying in particular condition \eqref{eq:K} for some $K>0$, there exists some $\alpha_0>0$ such that for every $\alpha$-bounded distribution $p(\delta)$ the distribution $p'(Q)$ of~\eqref{eq:pp} is $\beta$-bounded with $\beta$ as in~\eqref{eq:beta} and
\begin{equation}\label{eq:k}
k = \frac{1}{2(1+K)}.
\end{equation}
\end{prop}

\begin{proof}
Construct for each code in the family a graph with node set
\begin{equation}
\Gamma=\mathcal E\sqcup\mathcal Q.
\end{equation}
and such that two nodes $a$, $b$ are linked if
\begin{itemize}
\item
$a,b\in \mathcal E$ and $a$ and $b$ share an inner or outer endpoint,
\item
$a,b\in \mathcal Q$ and $a$ and $b$ are both in the same b-edge or in the same rg-plaquette, or
\item
$a\in \mathcal Q,b\in \mathcal E$ and $a$ is in a rg-plaquette that is an outer endpoint of $b$.
\end{itemize}
This definition is designed so that trivially (i) the family of graphs has bounded maximum degree and (ii) the following is satisfied:
given a connected component $\gamma_\mathrm c\sqcup Q_\mathrm c$ of a cluster $\gamma\sqcup Q$, where $\gamma_\mathrm c, \gamma\subseteq \mathcal E$ and $Q_\mathrm c, Q\subseteq \mathcal Q$,
and given $E\in\mathcal B$ with $\supp E=Q$,
\begin{align}
\delta\in\mathcal E_0 \quad &\Longrightarrow\quad \delta_\mathrm c\in \mathcal E_0,\label{eq:i1}\\
\synd E=\partial\delta\quad &\Longrightarrow\quad \synd E_\mathrm c=\partial\delta_\mathrm c,\label{eq:i2}
\end{align}
where $E_\mathrm c$ is the restriction of $E$ to $Q_\mathrm c$.

The result will follow by applying lemma~\ref{lem:bounded} to the above family of graphs, taking $\Gamma_1=\mathcal E$, $\Gamma_2=\mathcal Q$ and 
\begin{equation}
f(\delta)=(\delta+\delta_0)\cup\supp E_{\delta+\delta_0}.
\end{equation}
In this case \eqref{eq:connectivity} reads
\begin{equation}\label{eq:condconn}
2(1+K)|\delta\cap\delta_\mathrm c|\geq |\delta_\mathrm c|+|Q_\mathrm c|,
\end{equation}
where $\delta_\mathrm c\sqcup Q_\mathrm c$ is any connected component of $f(\delta)$.
Notice that
\begin{equation}
|\delta\cap \delta_\mathrm c|=|\delta_\mathrm c|-|\delta_0\cap \delta_\mathrm c|.
\end{equation}
According to~\eqref{eq:i1} $\delta_{\text c}\in\mathcal E_0$ and thus 
\begin{equation}
\delta+\delta_0'\in\mathcal E_0,\qquad\delta_0':=\delta_0+\delta_{\text c}.
\end{equation}
By the minimality of $\delta_0$~\eqref{eq:delta0} 
\begin{equation}\label{eq:byminimality}
0\leq|\delta_0'|-|\delta_0|=|\delta_\mathrm c|-2|\delta_\mathrm c\cap \delta_0|=2|\delta\cap \delta_\mathrm c|-|\delta_\mathrm c|.
\end{equation}
Thus
\begin{equation}\label{eq:condconn1}
2|\delta\cap \delta_\mathrm c|\geq |\delta_{\text c}|.
\end{equation}

If $E_\mathrm c$ is the restriction of $E_{\delta+\delta_0}$ to $Q_\mathrm c$, according to~\eqref{eq:i2} $\synd E_\mathrm c=\partial \delta_\mathrm c$.
Then
\begin{equation}
\synd E'=\synd E_{\delta+\delta_0},\qquad E':=E_{\delta+\delta_0} E_\mathrm c E_{\delta_\mathrm c},
\end{equation}
and
\begin{multline}
|\supp (E_{\delta+\delta_0}E_\mathrm c)|+|\supp E'_{\delta_\mathrm c}| \geq |\supp E'| \geq \\
\geq |\supp E_{\delta+\delta_0}|=|\supp (E_{\delta+\delta_0}E_\mathrm c)|+|Q_\mathrm c|.
\end{multline}
where  the second inequality is by the minimality of $E_{\delta+\delta_0}$~\eqref{eq:Egamma}.
Using~\eqref{eq:K}
\begin{equation}\label{eq:condconn2}
K|\delta_{\text c}|\geq |\supp E'_{\delta_\mathrm c}|\geq |Q_\mathrm c|.
\end{equation}
The inequalities~\eqref{eq:condconn1} and~\eqref{eq:condconn2} imply~\eqref{eq:condconn}.
\end{proof}

\bibliography{refs}

\begin{thebibliography}{25}
\expandafter\ifx\csname natexlab\endcsname\relax\def\natexlab#1{#1}\fi
\expandafter\ifx\csname bibnamefont\endcsname\relax
  \def\bibnamefont#1{#1}\fi
\expandafter\ifx\csname bibfnamefont\endcsname\relax
  \def\bibfnamefont#1{#1}\fi
\expandafter\ifx\csname citenamefont\endcsname\relax
  \def\citenamefont#1{#1}\fi
\expandafter\ifx\csname url\endcsname\relax
  \def\url#1{\texttt{#1}}\fi
\expandafter\ifx\csname urlprefix\endcsname\relax\def\urlprefix{URL }\fi
\providecommand{\bibinfo}[2]{#2}
\providecommand{\eprint}[2][]{\url{#2}}

\bibitem[{\citenamefont{Lidar and Brun~(editors)}(2013)}]{lidar:2013:quantum}
\bibinfo{author}{\bibfnamefont{D.}~\bibnamefont{Lidar}} \bibnamefont{and}
  \bibinfo{author}{\bibfnamefont{T.}~\bibnamefont{Brun~(editors)}},
  \emph{\bibinfo{title}{Quantum Error Correction}}
  (\bibinfo{publisher}{Cambridge University Press}, \bibinfo{address}{New
  York}, \bibinfo{year}{2013}).

\bibitem[{\citenamefont{Dennis et~al.}(2002)\citenamefont{Dennis, Kitaev,
  Landahl, and Preskill}}]{dennis:2002:TQM}
\bibinfo{author}{\bibfnamefont{E.}~\bibnamefont{Dennis}},
  \bibinfo{author}{\bibfnamefont{A.}~\bibnamefont{Kitaev}},
  \bibinfo{author}{\bibfnamefont{A.}~\bibnamefont{Landahl}}, \bibnamefont{and}
  \bibinfo{author}{\bibfnamefont{J.}~\bibnamefont{Preskill}},
  \bibinfo{journal}{J. Math. Phys.} \textbf{\bibinfo{volume}{43}},
  \bibinfo{pages}{4452} (\bibinfo{year}{2002}).

\bibitem[{\citenamefont{Kitaev}(2003)}]{kitaev:2003:ftanyons}
\bibinfo{author}{\bibfnamefont{A.}~\bibnamefont{Kitaev}},
  \bibinfo{journal}{Ann. Phys.} \textbf{\bibinfo{volume}{303}},
  \bibinfo{pages}{2} (\bibinfo{year}{2003}).

\bibitem[{\citenamefont{Bravyi and
  K{\"o}nig}(2013)}]{bravyi:2013:classification}
\bibinfo{author}{\bibfnamefont{S.}~\bibnamefont{Bravyi}} \bibnamefont{and}
  \bibinfo{author}{\bibfnamefont{R.}~\bibnamefont{K{\"o}nig}},
  \bibinfo{journal}{Phys. Rev. Lett.} \textbf{\bibinfo{volume}{110}},
  \bibinfo{pages}{170503} (\bibinfo{year}{2013}).

\bibitem[{\citenamefont{Bombin and
  Martin-Delgado}(2007{\natexlab{a}})}]{bombin:2007:3DCC}
\bibinfo{author}{\bibfnamefont{H.}~\bibnamefont{Bombin}} \bibnamefont{and}
  \bibinfo{author}{\bibfnamefont{M.~A.} \bibnamefont{Martin-Delgado}},
  \bibinfo{journal}{Phys. Rev. Lett.} \textbf{\bibinfo{volume}{98}},
  \bibinfo{pages}{160502} (\bibinfo{year}{2007}{\natexlab{a}}).

\bibitem[{\citenamefont{Bombin and Martin-Delgado}(2006)}]{bombin:2006:2DCC}
\bibinfo{author}{\bibfnamefont{H.}~\bibnamefont{Bombin}} \bibnamefont{and}
  \bibinfo{author}{\bibfnamefont{M.~A.} \bibnamefont{Martin-Delgado}},
  \bibinfo{journal}{Phys. Rev. Lett.} \textbf{\bibinfo{volume}{97}},
  \bibinfo{pages}{180501} (\bibinfo{year}{2006}).

\bibitem[{\citenamefont{Nigg~\emph{et al.}}(2014)}]{nigg:2014:experimental}
\bibinfo{author}{\bibfnamefont{D.}~\bibnamefont{Nigg~\emph{et al.}}},
  \bibinfo{journal}{Science} \textbf{\bibinfo{volume}{345}},
  \bibinfo{pages}{302} (\bibinfo{year}{2014}).

\bibitem[{\citenamefont{Pastawski and Yoshida}(2015)}]{pastawski:2015:fault}
\bibinfo{author}{\bibfnamefont{F.}~\bibnamefont{Pastawski}} \bibnamefont{and}
  \bibinfo{author}{\bibfnamefont{B.}~\bibnamefont{Yoshida}},
  \bibinfo{journal}{Physical Review A} \textbf{\bibinfo{volume}{91}},
  \bibinfo{pages}{012305} (\bibinfo{year}{2015}).

\bibitem[{\citenamefont{Bravyi and Kitaev}(2005)}]{bravyi:2005:universal}
\bibinfo{author}{\bibfnamefont{S.}~\bibnamefont{Bravyi}} \bibnamefont{and}
  \bibinfo{author}{\bibfnamefont{A.}~\bibnamefont{Kitaev}},
  \bibinfo{journal}{Phys. Rev. A} \textbf{\bibinfo{volume}{71}},
  \bibinfo{pages}{22316} (\bibinfo{year}{2005}).

\bibitem[{\citenamefont{Bombin}(2015{\natexlab{a}})}]{bombin:2015:gauge}
\bibinfo{author}{\bibfnamefont{H.}~\bibnamefont{Bombin}}, \bibinfo{journal}{New
  Journal of Physics} \textbf{\bibinfo{volume}{17}}, \bibinfo{pages}{083002}
  (\bibinfo{year}{2015}{\natexlab{a}}).

\bibitem[{\citenamefont{Bombin}(2015{\natexlab{b}})}]{bombin:2015:single-shot}
\bibinfo{author}{\bibfnamefont{H.}~\bibnamefont{Bombin}},
  \bibinfo{journal}{Phys. Rev. X} \textbf{\bibinfo{volume}{5}},
  \bibinfo{pages}{031043} (\bibinfo{year}{2015}{\natexlab{b}}).

\bibitem[{\citenamefont{Brown et~al.}(2015)\citenamefont{Brown, Nickerson, and
  Browne}}]{brown:2015:fault}
\bibinfo{author}{\bibfnamefont{B.~J.} \bibnamefont{Brown}},
  \bibinfo{author}{\bibfnamefont{N.~H.} \bibnamefont{Nickerson}},
  \bibnamefont{and} \bibinfo{author}{\bibfnamefont{D.~E.}
  \bibnamefont{Browne}}, \bibinfo{journal}{arXiv:1503.08217}
  (\bibinfo{year}{2015}).

\bibitem[{\citenamefont{Poulin}(2005)}]{poulin:2005:stabilizer}
\bibinfo{author}{\bibfnamefont{D.}~\bibnamefont{Poulin}},
  \bibinfo{journal}{Phys. Rev. Lett.} \textbf{\bibinfo{volume}{95}},
  \bibinfo{pages}{230504} (\bibinfo{year}{2005}).

\bibitem[{\citenamefont{Paetznick and
  Reichardt}(2013)}]{paetznick:2013:universal}
\bibinfo{author}{\bibfnamefont{A.}~\bibnamefont{Paetznick}} \bibnamefont{and}
  \bibinfo{author}{\bibfnamefont{B.~W.} \bibnamefont{Reichardt}},
  \bibinfo{journal}{Phys. Rev. Lett.} \textbf{\bibinfo{volume}{111}},
  \bibinfo{pages}{90505} (\bibinfo{year}{2013}).

\bibitem[{\citenamefont{Bombin and
  Martin-Delgado}(2007{\natexlab{b}})}]{bombin:2007:branyons}
\bibinfo{author}{\bibfnamefont{H.}~\bibnamefont{Bombin}} \bibnamefont{and}
  \bibinfo{author}{\bibfnamefont{M.~A.} \bibnamefont{Martin-Delgado}},
  \bibinfo{journal}{Phys. Rev. B} \textbf{\bibinfo{volume}{75}},
  \bibinfo{pages}{75103} (\bibinfo{year}{2007}{\natexlab{b}}).

\bibitem[{\citenamefont{Kubica and Beverland}(2015)}]{kubica:2015:universal}
\bibinfo{author}{\bibfnamefont{A.}~\bibnamefont{Kubica}} \bibnamefont{and}
  \bibinfo{author}{\bibfnamefont{M.~E.} \bibnamefont{Beverland}},
  \bibinfo{journal}{Physical Review A} \textbf{\bibinfo{volume}{91}},
  \bibinfo{pages}{032330} (\bibinfo{year}{2015}).

\bibitem[{\citenamefont{Katzgraber et~al.}(2009)\citenamefont{Katzgraber,
  Bombin, and Martin-Delgado}}]{katzgraber:2009:cc}
\bibinfo{author}{\bibfnamefont{H.}~\bibnamefont{Katzgraber}},
  \bibinfo{author}{\bibfnamefont{H.}~\bibnamefont{Bombin}}, \bibnamefont{and}
  \bibinfo{author}{\bibfnamefont{M.}~\bibnamefont{Martin-Delgado}},
  \bibinfo{journal}{Phys. Rev. Lett.} \textbf{\bibinfo{volume}{103}},
  \bibinfo{pages}{90501} (\bibinfo{year}{2009}).

\bibitem[{\citenamefont{Katzgraber et~al.}(2010)\citenamefont{Katzgraber,
  Bombin, Andrist, and Martin-Delgado}}]{katzgraber:2009:unionjack}
\bibinfo{author}{\bibfnamefont{H.}~\bibnamefont{Katzgraber}},
  \bibinfo{author}{\bibfnamefont{H.}~\bibnamefont{Bombin}},
  \bibinfo{author}{\bibfnamefont{R.}~\bibnamefont{Andrist}}, \bibnamefont{and}
  \bibinfo{author}{\bibfnamefont{M.}~\bibnamefont{Martin-Delgado}},
  \bibinfo{journal}{Phys. Rev. A (2010)} \textbf{\bibinfo{volume}{81}},
  \bibinfo{pages}{12319} (\bibinfo{year}{2010}).

\bibitem[{\citenamefont{Andrist et~al.}(2010)\citenamefont{Andrist, Katzgraber,
  Bombin, and Martin-Delgado}}]{andrist:2010:tricolored}
\bibinfo{author}{\bibfnamefont{R.}~\bibnamefont{Andrist}},
  \bibinfo{author}{\bibfnamefont{H.}~\bibnamefont{Katzgraber}},
  \bibinfo{author}{\bibfnamefont{H.}~\bibnamefont{Bombin}}, \bibnamefont{and}
  \bibinfo{author}{\bibfnamefont{M.}~\bibnamefont{Martin-Delgado}},
  \bibinfo{journal}{New J. Phys.} \textbf{\bibinfo{volume}{13}},
  \bibinfo{pages}{083006} (\bibinfo{year}{2010}).

\bibitem[{\citenamefont{Landahl et~al.}(2011)\citenamefont{Landahl, Anderson,
  and Rice}}]{landahl:2011:fault}
\bibinfo{author}{\bibfnamefont{A.~J.} \bibnamefont{Landahl}},
  \bibinfo{author}{\bibfnamefont{J.~T.} \bibnamefont{Anderson}},
  \bibnamefont{and} \bibinfo{author}{\bibfnamefont{P.~R.} \bibnamefont{Rice}},
  \bibinfo{journal}{arXiv preprint arXiv:1108.5738}  (\bibinfo{year}{2011}).

\bibitem[{\citenamefont{Bombin et~al.}(2012)\citenamefont{Bombin, Andrist,
  Ohzeki, Katzgraber, and Martin-Delgado}}]{bombin:2012:strong}
\bibinfo{author}{\bibfnamefont{H.}~\bibnamefont{Bombin}},
  \bibinfo{author}{\bibfnamefont{R.}~\bibnamefont{Andrist}},
  \bibinfo{author}{\bibfnamefont{M.}~\bibnamefont{Ohzeki}},
  \bibinfo{author}{\bibfnamefont{H.}~\bibnamefont{Katzgraber}},
  \bibnamefont{and}
  \bibinfo{author}{\bibfnamefont{M.}~\bibnamefont{Martin-Delgado}},
  \bibinfo{journal}{Physical Review X} \textbf{\bibinfo{volume}{2}},
  \bibinfo{pages}{21004} (\bibinfo{year}{2012}).

\bibitem[{\citenamefont{Sarvepalli and
  Raussendorf}(2012)}]{sarvepalli:2012:efficient}
\bibinfo{author}{\bibfnamefont{P.}~\bibnamefont{Sarvepalli}} \bibnamefont{and}
  \bibinfo{author}{\bibfnamefont{R.}~\bibnamefont{Raussendorf}},
  \bibinfo{journal}{Physical Review A} \textbf{\bibinfo{volume}{85}},
  \bibinfo{pages}{022317} (\bibinfo{year}{2012}).

\bibitem[{\citenamefont{Delfosse}(2014)}]{delfosse:2014:decoding}
\bibinfo{author}{\bibfnamefont{N.}~\bibnamefont{Delfosse}},
  \bibinfo{journal}{Physical Review A} \textbf{\bibinfo{volume}{89}},
  \bibinfo{pages}{012317} (\bibinfo{year}{2014}).

\bibitem[{\citenamefont{Stephens}(2014)}]{stephens:2014:efficient}
\bibinfo{author}{\bibfnamefont{A.~M.} \bibnamefont{Stephens}},
  \bibinfo{journal}{arXiv:1402.3037}  (\bibinfo{year}{2014}).

\bibitem[{\citenamefont{Anderson et~al.}(2014)\citenamefont{Anderson,
  Duclos-Cianci, and Poulin}}]{anderson:2014:fault}
\bibinfo{author}{\bibfnamefont{J.~T.} \bibnamefont{Anderson}},
  \bibinfo{author}{\bibfnamefont{G.}~\bibnamefont{Duclos-Cianci}},
  \bibnamefont{and} \bibinfo{author}{\bibfnamefont{D.}~\bibnamefont{Poulin}},
  \bibinfo{journal}{Physical review letters} \textbf{\bibinfo{volume}{113}},
  \bibinfo{pages}{080501} (\bibinfo{year}{2014}).

\end{thebibliography}

\end{document}